\documentclass[12pt]{amsart}
\usepackage{amsmath,mathrsfs,bm,amssymb,bbm}
\usepackage{hyperref}
\usepackage{latexsym}
\newtheorem{thm}{Theorem}[section]

 \newtheorem{lem}[thm]{Lemma}
 \newtheorem{prop}[thm]{Proposition}
 \theoremstyle{definition}

 \theoremstyle{remark}
 \newtheorem{rem}[thm]{Remark}
 \newtheorem{exam}[thm]{Example}
 \numberwithin{equation}{section}
\newcommand{\cB}{\mathcal{B}}
\newcommand{\cC}{\mathcal{C}}
\newcommand{\cD}{\mathcal{D}}

\newcommand{\cF}{\mathcal{F}}
\newcommand{\cG}{\mathcal{G}}
\newcommand{\cH}{\mathcal{H}}

\newcommand{\cX}{\mathcal{X}}
\newcommand{\sA}{\mathscr{A}}

\newcommand{\sC}{\mathscr{C}}
\newcommand{\sF}{\mathscr{F}}

\newcommand{\sI}{\mathscr{I}}

\newcommand{\sL}{\mathscr{L}}
\newcommand{\RR}{\mathbb{R}}
\newcommand{\CC}{\mathbb{C}}
\newcommand{\EE}{\mathbb{E}}
\newcommand{\NN}{\mathbb{N}}

\newcommand{\PP}{\mathbb{P}}
\newcommand{\eps}{\varepsilon}
\newcommand{\om}{\omega}

\newcommand{\one}{{\mathbbm{1}}}

\newcommand{\half}{\frac{1}{2}}
\newcommand{\rd}{{\rm d}}
\newcommand{\lk}{\left(}
\newcommand{\rk}{\right)}

\newcommand{\add}{a^\dagger}

\newcommand{\hf}{H_{\rm f}}
\newcommand{\hp}{H_{\rm p}}
\newcommand{\hn}{H}
\newcommand{\gr}{\Phi_{\rm g}}
\newcommand{\eg}{E_{\rm g}}
\begin{document}

\title[Two-sided bounds on ground states in the Nelson model]{Two-sided bounds on the point-wise spatial decay of ground states in the 
renormalized Nelson model with confining potentials}

\author{Fumio Hiroshima and Oliver Matte}

\date{\today}

\address{Fumio Hiroshima, Faculty of Mathematics, Kyushu University, Fukuoka Motooka 744, 819-0395, Japan}
\email{hiroshima@math.kyushu-u.ac.jp}

\address{Oliver Matte, Institut for Matematiske Fag, Aalborg Universitet, Thomas Manns Vej 23, 9220 Aalborg {\O}, Denmark}
\email{oliver@math.aau.dk}

\maketitle

\setlength{\baselineskip}{13pt}

\begin{abstract}
\noindent
We study the renormalized Nelson model for a scalar matter particle in a continuous confining potential
interacting with a possibly massless quantized radiation field.
When the radiation field is massless we impose a mild infrared regularization ensuring that the Nelson Hamiltonian
has a non-degenerate ground state in all considered cases. Employing Feynman--Kac representations, 
we derive lower bounds on the point-wise spatial decay of the partial
Fock space norms of ground state eigenvectors. Here the exponential rate function governing the decay is given by the Agmon distance familiar from
the analysis of Schr\"{o}dinger operators. For a large class of confining potentials, our lower bounds on the decay
of ground state eigenvectors match asymptotically with the upper bounds implied by previous work of the present authors.
\end{abstract}


\section{Introduction and main result}\label{secintro}

\subsection{General introduction}\label{ssecgenintro}

\noindent
In this article we discuss upper and lower bounds on the point-wise spatial decay of 
ground state eigenvectors of the renormalized Nelson Hamiltonian in quantum field theory.
This Hamiltonian was constructed by E.~Nelson \cite{Nelson1964} in 1964. 
It generates the dynamics of a non-relativistic scalar quantum particle interacting with a quantized radiation field comprised
of scalar bosons. The crucial feature of Nelson's model is that, by an energy renormalization,
one can remove an artificial ultraviolet cutoff imposed in the matter-radiation interaction term for a start to render 
heuristic physical expressions mathematically well-defined. Due to its more involved construction---see 
also \cite{GriesemerWuensch2018,GHL14,LampartSchmidt2019,MatteMoeller2018,Posilicano2020,Schmidt2021} 
for improvements and approaches alternative to \cite{Nelson1964}---the spectral and probabilistic analysis of the renormalized Nelson model becomes more 
challenging compared to similar models with ultraviolet regularized matter-radiation interactions.

In a previous article \cite{HiroshimaMatte2022}, the present authors proved the existence of ground state eigenvectors of 
renormalized Nelson Hamiltonians with Kato-decomposable potentials assuming a binding condition as well as a mild
infrared regularity condition in case the boson mass is zero. For massless boson fields the latter infrared condition was also shown to be necessary.
Furthermore, upper bounds on the point-wise spatial decay of the partial Fock space norms of ground state eigenvectors were
obtained in \cite{HiroshimaMatte2022}, for potentials leading to a finite localization threshold as well as for a class of confining potentials
increasing at least polynomially at infinity. The ground state eigenvectors found in \cite{HiroshimaMatte2022} have representatives depending continuously on
the position variable of the matter particle, and it is the point-wise decay of these representatives that we address here and henceforth.

In the present article we complement these results by deriving lower bounds on the point-wise decay of ground state eigenvectors for continuous confining potentials.
As many of the results in \cite{HiroshimaMatte2022}, the proofs of our lower bounds are based on Feynman--Kac representations of the semigroup
generated by the renormalized Nelson Hamiltonian \cite{GHL14,MatteMoeller2018}. We shall proceed along the lines developed by
R.~Carmona and B.~Simon in the theory of Schr\"{o}dinger operators \cite{CS81}, adding new arguments to deal with the radiation field. 
In particular, we shall estimate the decay in terms of the Agmon distance employed in \cite{CS81}, i.e., the geodesic distance to the origin
arising when $\RR^3$ is equipped with the Riemannian metric tensor $g_x(v,w):=2V(x)v\cdot w$ at every $x\in\RR^3$; here $v,w\in\RR^3$ are tangent vectors
and $V$ is the (positive) potential. $g_x$ is a tunneling region analog of the Jacobi metric in classical mechanics. In fact, in the lower bounds
valid for the most general potentials considered here, $V(x)$ has to be replaced by $V_\circ(x):=\sup_{|y|\leq1}V(x+y)$ in the definition of $g_x$, but in many examples
one can return to the original $V$.
The Agmon distance is encountered as an exponential rate function by minimizing the classical action functional associated with $-V_\circ$,
which in turn naturally appears through the Girsanov theorem when the integration over paths in the Feynman--Kac formula for ground states 
is restricted to an integration over a neighborhood of an action minimizing geodesic~\cite[Section~2]{CS81}. 
As mentioned in \cite{CS81}, Varadhan \cite[(2.8)]{var67} 
seems to be the first to use similar geodesic ideas in estimating transition probabilities of diffusion processes.

There is a wealth of literature on decay estimates on bound states of Schr\"odinger type operators; here we only refer to
the books~\cite{Agmon1982,HS96} and mention that probabilistic approaches alternative to \cite{CS81} can be found in \cite{AS82, car78,CMS90}.
Notably, Agmon's distance has generalizations to $N$-body Schr\"{o}dinger operators \cite{Agmon1982,CS81}.

Prior to \cite{HiroshimaMatte2022},
point-wise upper bounds on the fall-off of ground state eigenvectors and more general elements of spectral subspaces
in ultraviolet regularized models of non-relativistic quantum field theory appeared in the articles \cite{BHLMS01,HiHi2010,Hiroshima2019,Matte2016},
which all rely on probabilistic techniques. Hence, the main novelty of \cite{HiroshimaMatte2022} was that no 
ultraviolet regularizations were necessary any longer. The recent preprint \cite{GriesemerKussmaul2024}
contains short proofs of $L^2$ to $L^\infty$ bounds for ultraviolet regularized models based on subsolution estimates, 
which permit to turn $L^2$-exponential localization estimates for ground states \cite{BFS1998b,Griesemer2004,Panati2009} into point-wise bounds.
The point-wise bounds in \cite{Matte2016} actually apply to any finite number of matter particles.

Also lower point-wise bounds on the partial Fock space norms of ground state eigenvectors in 
models involving quantum fields appeared earlier: for the ultraviolet regularized Nelson model such a bound can be found in \cite[Theorem 2.93]{HL20}, 
and for the semi-relativistic Pauli--Fierz model in \cite{hir14}. Hence, again, the first novel aspect of the results obtained here
is the absence of ultraviolet regularity conditions. In this regard, a technical key input from \cite{MatteMoeller2018} are exponential moment estimates
on processes appearing in the Feynman--Kac formula, whose right hand sides grow at most log-linearly in time.
Besides, we find the first lower bounds on eigenvectors in models involving quantized fields that are expressed in terms of the 
aforementioned Agmon distance, and for a large class of continuous, confining potentials we demonstrate that point-wise upper bounds 
involving the same Agmon distance can be read off from \cite{HiroshimaMatte2022}.

In particular, the latter two-sided bounds match asymptotically and reveal the same leading order behavior
as found in \cite{CS81} for Schr\"{o}dinger operators; see Remark~\ref{remasympmatch}. This is because, for confining potentials,
the eikonal equation determining the Agmon distance near infinity, $|\nabla \varrho|^2=2V$, is the same for Schr\"{o}dinger operators and
Nelson Hamiltonians. For decaying potentials, the finite binding energy of a ground state replaces $V$ in the eikonal equation, and
since binding energies are strictly increased by adding the field theoretic terms to the Schr\"{o}dinger operator \cite{KoenenbergMatte2013}, this results in an enhanced
decay of ground state eigenvectors as reflected by upper bounds in the $L^2$-sense in \cite{Griesemer2004,HiroshimaMatte2022,Panati2009}
and point-wise upper bounds in \cite{GriesemerKussmaul2024,HiroshimaMatte2022,Matte2016}.
Deriving asymptotically matching upper and lower bounds for decaying potentials in non-relativistic quantum field theory seems to be a more difficult problem.

We remark that other types of ground state localizations can be studied at every point in space, like exponential localization at any rate
in the boson number or Gaussian localization with respect to bosonic field variables \cite{HiroshimaMatte2022}; see the brief
summaries in Remarks~\ref{remotherloc1} and~\ref{remotherloc2}.

In the following two subsections we introduce the Nelson model and present our main results.
In preparation for the proof of our lower bounds, we explain the Feynman--Kac formula for the renormalized Nelson model as well as
some of its implications relevant to us in Section~\ref{secFK}. In that section we also extract upper point-wise bounds involving
Agmon distances from \cite{HiroshimaMatte2022}, and we give a proof of the continuity of functions in the range
of semigroup elements, which is a good deal shorter than the proof of a stronger equicontinuity statement in \cite{Matte2016}.
Section~\ref{secproof} contains the proof of our lower bounds.
Some well-known results on the Agmon distance of particular importance to us are explained in Appendix~\ref{appdist}.


\subsection{The Nelson Hamiltonian}\label{ssecNelsonHam}

With $\cF$ denoting the bosonic Fock space over $L^2(\RR^3)$ (see \eqref{defFock}), the total Hilbert space for the matter-radiation system treated here is 
\begin{align*}
\cH:=L^2(\RR^3,\cF).
\end{align*} 
While we eventually will restrict our attention to continuous confining exterior potentials
$V:\RR^3\to\RR$, we shall for a start and in some technical results merely assume $V$ to be Kato-decomposable.
That is, $V$ is measurable, $V_-:=\max\{-V,0\}$ satisfies
\begin{align}\label{Katoclass}
\lim_{r\downarrow0}\sup_{x\in\RR^3}\int_{\RR^3}\one_{[0,r)}(|x-y|)\frac{V_-(y)}{|x-y|}\rd y&=0,
\end{align}
and the positive part $V_+:=V+V_-$ satisfies this condition locally, i.e., \eqref{Katoclass} holds with
$\chi_KV_+$ put in place of $V_-$ for any compact $K\subset\RR^3$. Under these conditions we can define
a Schr\"{o}dinger operator $\hp$ (particle Hamiltonian) acting in $\cH$ as the semi-bounded form sum
\begin{align*}
\hp:=-\half\Delta\dot{+}V.
\end{align*}
Here both $\Delta$ and $V$ act on $\cF$-valued functions on $\RR^3$ in the canonical way.

Turning to the quantum field, the bosons, here described in momentum space, have the relativistic dispersion relation
\begin{align*}
\om(k)&:=\sqrt{|k|^2+\nu^2},\quad k\in\RR^3,\quad\text{for some $\nu\ge0$.}
\end{align*}
In what follows $\hf:=\rd\Gamma(\omega)$ denotes the differential second quantization of the multiplication operator induced by $\om$ (see \eqref{defhf}).
The heuristic interaction kernel for the matter-radiation interaction is given by choosing $\Lambda=\infty$ and $\eta=1$ in
\begin{align}\label{defxexk}
v_{\Lambda,x}(k)&:= g\frac{e_x(k)\eta(k)}{\sqrt{\om(k)}}\one_{[0,\Lambda)}(|k|),\quad e_x(k):=e^{-i k\cdot x},\qquad k,x\in\RR^3,\,\Lambda\in(0,\infty].
\end{align}
Here the value of the coupling constant $g\in\RR$ can be chosen arbitrarily. The measurable function $\eta:\RR^3\to\RR$ is assumed to satisfy
\begin{align*}
0\leq\eta(k)\leq1,\quad \eta(k)=\eta(-k), \quad k\in\RR^3.
\end{align*}
It implements an infrared regularization if necessary, see \eqref{IRcond} below, and ensures that ultraviolet regularized versions of Nelson's model
are covered by our results as well. 
Whenever $v_{\Lambda,x}$ is square-integrable, the corresponding field operator $\varphi(v_{\Lambda,x})$ acting in $\cF$ is well-defined (see \eqref{defavp}).
Hence, we can define ultraviolet regularized Nelson Hamiltonians $H_{\Lambda}$ with $\Lambda\in(0,\infty)$ setting
\begin{align*}
(H_{\Lambda}\Psi)(x)&:= (\hp\Psi)(x)+\hf\Psi(x)+\varphi(v_{\Lambda,x})\Psi(x)+E_{\Lambda}\Psi(x),\quad\text{a.e. $x\in\RR^3$,}
\end{align*}
for all $\Psi\in\cD(H_{\Lambda})=\cD(\hp)\cap L^2(\RR^3,\cD(\hf))$.
Here and henceforth, $\cD(\cdot)$ denotes domains of definition, and we added right away the renormalization energies
\begin{align}\label{defEren}
E_{\Lambda}:=g^2\int_{\RR^3}\frac{\eta(k)^2\one_{[0,\Lambda)}(|k|)}{\om(k)(\om(k)+|k|^2/2)}\rd k,\quad \Lambda\in(0,\infty).
\end{align}
Note that $E_{\Lambda}$ diverges logarithmically as $\Lambda\to\infty$ when $g\not=0$ and $\eta(k)=1$ for large $|k|$.
By the relative bound \eqref{rbvp} and the Kato-Rellich theorem $H_{\Lambda}$ is selfadjoint for every $\Lambda\in(0,\infty)$.
Nelson's result \cite{Nelson1964} (and later improvements covering massless bosons and revealing stronger
convergence properties \cite{Ammari2000,GriesemerWuensch2018,GHL14,HiroshimaMatte2022,LampartSchmidt2019,MatteMoeller2018,Posilicano2020,Schmidt2021}) 
ensures the existence of the renormalized Nelson Hamiltonian
\begin{align}\label{N3}
\hn:=\lim_{\Lambda \to\infty} H_{\Lambda}  \quad\text{in the norm resolvent sense.}
\end{align}
$H$ is semi-bounded from below and we denote the minimum of its spectrum by
\begin{align*}
E_{\rm g}&:=\inf\sigma(\hn).
\end{align*}
Our analysis departs from the next proposition, which is a special case of Theorem~1.4 in our previous article \cite{HiroshimaMatte2022}; 
the existence of ground states of Nelson Hamiltonians with ultraviolet cutoffs has been shown earlier in \cite{ger00,Panati2009,spo98}. 
Note that \eqref{IRcond} is always fulfilled if $\nu>0$.

\begin{prop}\label{propexgr}
Assume that $V(x)\to\infty$ as $|x|\to\infty$.
Further, assume the following infrared regularity condition to hold:
\begin{align}\label{IRcond}
\int_{\RR^3}\one_{(0,1)}(|k|)\frac{\eta(k)^2}{\omega(k)^3}\rd k<\infty.
\end{align}
Then $E_{\rm g}$ is a non-degenerate eigenvalue of $H$.
\end{prop}

Any normalized eigenvector of $H$ corresponding to $E_{\rm g}$ will be called a ground state eigenvector and denoted by $\gr$ in what follows.

In fact, the above proposition holds for all Kato-decomposable $V$ satisfying a binding condition,
which trivially follows for $V$ going to infinity as $|x|\to\infty$; see \cite[\textsection3.3]{HiroshimaMatte2022}. 
The infrared condition \eqref{IRcond} actually is optimal, as one can prove absence of ground states for arbitrary Kato-decomposable 
$V$ when \eqref{IRcond} does not hold \cite[Theorem~1.5]{HiroshimaMatte2022}. 


\subsection{Spatial decay of ground state eigenvectors}\label{ssecspatdecay}

Let us now describe our main results. In the whole Section~\ref{ssecspatdecay} we assume that 
\begin{align}\label{hypV}
\text{$V$ is continuous, $V\geq1$ on $\RR^3$ and} \ \
V(x)\xrightarrow{\;\;|x|\to\infty\;\;}\infty.
\end{align} 

First, we introduce the exponential rate functions governing the decay of ground state eigenvectors. 
To this end we introduce spaces of Lipschitz continuous paths
\begin{align}\label{defcCT}
\cC_T(x,y):=\{\gamma\in {\rm Lip}([0,T],\RR^3)\mid \gamma(0)=y, \gamma(T)=x\},\quad x,y\in\RR^3,\,T>0,
\end{align}
as well as the length functional, defined on the union of all these path spaces,
\begin{align}\label{defsL}
&\sL_V(\gamma):=\int_0^T\sqrt{2V(\gamma(s))}|\dot \gamma(s)|\rd s,\quad \gamma\in\cC_T(x,y).
\end{align}
Since every $\gamma\in \cC_T(x,y)$ is a.e. differentiable with an essentially bounded derivative, $\sL_V(\gamma)$ is well-defined,
and the value $\sL_V(\gamma)$ is manifestly invariant under Lipschitz re-parametrizations of $\gamma$. 
It is therefore natural to define the geodesic distances
\begin{align}\label{defvrx}
d(x,y)&:=\inf_{\gamma\in\cC_1(x,y)}\sL_V(\gamma),\quad \varrho(x):=d(x,0),\qquad x,y\in\RR^3.
\end{align}
In studies of eigenfunctions of Schr\"{o}dinger operators, $d$ is commonly called Agmon distance \cite{CS81}. 
We stick to this nomenclature, as we will adapt arguments of \cite{CS81} to Nelson's model.

The Agmon distance to the origin $\varrho$ appears in our decay estimates.
It is locally Lipschitz continuous, in particular differentiable a.e. on $\RR^3\setminus\{0\}$, and
it satisfies the eikonal equation
\begin{align}\label{eikonaleq}
\frac{1}{2}|\nabla\rho(x)|^2=V(x),\quad\text{a.e. $x\in\RR^3\setminus\{0\}$;}
\end{align} 
see Lemma~\ref{lemeikonal} for an explanation as to why these facts follows from (e.g.) \cite[Theorem~5.2]{Lions1982}.

As mentioned in the general introduction, we already derived upper bounds on the $L^2$-localization and the point-wise decay of ground state eigenvectors
and more general elements of spectral subspaces of $H$ in \cite{HiroshimaMatte2022}. Combining these results with \eqref{eikonaleq}, we arrive at
the following proposition. The technical growth restriction \eqref{extraV} is convenient when dealing with the issue that $\varrho$ might not be
globally Lipschitz continuous, which complicates the application of exponentially weighted $L^2$ to $L^\infty$ bounds on the semi-group $e^{-TH}$. 

\begin{prop}\label{propub}
Assume that $V$ satisfies \eqref{hypV}. Let $\eps>0$ and $\lambda\geq E_{\rm g}$. Then the range of the spectral projection 
$\one_{(-\infty,\lambda]}(H)$ is contained in the domain of the maximal multiplication operator $e^{(1-\eps)\varrho}$.
Let $\Psi$ be any normalized element in the range of $\one_{(-\infty,\lambda]}(H)$. Then $\Psi$ has a unique continuous representative $\Psi(\cdot)$.
Assume in addition that, for every $\delta>0$, there exists $D_\delta\in[1,\infty)$ such that
\begin{align}\label{extraV}
V(x)&\leq D_\delta e^{\delta|x| },\quad x\in\RR^3.
\end{align}
Then there exists $C_{\eps,\lambda}\in(0,\infty)$, also depending on $V$, $g\eta$ and $\nu$, such that the
continuous representative of $\Psi$ satisfies
\begin{align}\label{pwub}
\|\Psi(x)\|_\cF&\leq C_{\eps,\lambda}e^{-(1-\eps)\varrho(x)},\quad x\in\RR^3.
\end{align}
\end{prop}

\begin{proof}
In Section~\ref{ssecubproof} 
we explain how the claimed localization properties can be extracted from \cite{HiroshimaMatte2022}.
 The existence of continuous representatives is addressed in Section~\ref{sseccont}.
\end{proof}

Our main result is a lower bound for ground state eigenvectors involving a modified Agmon distance
arising when the techniques of \cite{CS81} are applied, namely
\begin{align}\label{defVplus}
\varrho_{\circ}(x)&:=\inf_{\gamma\in\cC_1(x,0)}\sL_{V_{\circ}}(\gamma),\quad x\in\RR^3,\quad\text{with}\quad V_{\circ}(x):=\sup_{|x-y|\leq1}V(y),\quad x\in\RR^3.
\end{align}
Our lower bound matches asymptotically with \eqref{pwub} (see Remark~\ref{remasympmatch}) and $\varrho$ can be put in place of 
$\varrho_\circ$ in the next theorem, whenever
\begin{align}\label{asympmatch}
\frac{\varrho_{\circ}(x)}{\varrho(x)}\xrightarrow{\;\;|x|\to\infty\;\;}1.
\end{align}
We denote by $\Omega$ the vacuum vector in $\cF$ (see \eqref{vacuum}), so that the first bound in \eqref{pwlb} is just
the Cauchy--Schwarz inequality for the scalar product on $\cF$.

\begin{thm}\label{mainthmvr}
Assume that $V$ satisfies \eqref{hypV} and \eqref{IRcond} holds. Let $\eps>0$.
Then there exists $c_{\eps}\in(0,\infty)$, also depending on $V$, $g\eta$ and $\nu$, such that
the unique continuous representative $\gr(\cdot)$ of the (normalized) ground state eigenvector $\gr$ satisfies
\begin{align}\label{pwlb}
\|\gr(x)\|_\cF\geq|(\Omega,\gr(x))_{\cF}|&\geq c_{\eps}e^{-(1+\eps)\varrho_{\circ}(x)},\quad x\in\RR^3.
\end{align}
\end{thm}

\begin{proof}
The proof of this theorem can be found at the end of Section~\ref{s4}.
\end{proof}

\begin{rem}\label{remasympmatch}
Analogously to \cite{CS81}, if \eqref{asympmatch} is satisfied and \eqref{pwub} and \eqref{pwlb} hold for all $\eps>0$, then we obtain  
\begin{align*}
-\frac{\ln(\|\gr(x)\|_{\cF})}{\varrho(x)}\xrightarrow{\;\;|x|\to\infty\;\;}1.
\end{align*} 
This showcases the same leading asymptotic behavior as for ground state eigenfunctions of Schr\"{o}dinger operators
with potentials satisfying \eqref{hypV} and \eqref{asympmatch}, \cite{CS81}. 
\end{rem}

 \begin{exam}
 Besides \eqref{hypV}, assume that $V$ is differentiable on the complement of some ball in $\RR^3$ and
 $|\nabla V(x)|/V(x)\to0$ as $|x|\to\infty$. As noted in \cite{CS81}, \eqref{asympmatch} holds in this situation. 
 \end{exam}

\begin{exam} 
Assume that $V\geq1$ is continuous and obeys the bounds
\begin{align}\label{hypVexam}
\frac{a^2}{2}|x|^{2n}-b\leq V(x)\leq \frac{A^2}{2}|x|^{2m}+B,\quad x\in\RR^3,
\end{align}
for some $m\geq n>0$, $a,A>0$ and $b,B\geq 0$. Let $\eps>0$. Then we find $C_\eps\geq c_\eps>0$, also depending
on all parameters in \eqref{hypVexam} as well as on $g\eta$ and $\nu$, such that
\begin{align}
\label{oliver}
c_\eps e^{-(1+\eps)\frac{A}{m+1}|x|^{m+1}}\leq
\|\gr(x)\|_\cF\leq C_\eps e^{-(1-\eps)\frac{a}{n+1}|x|^{n+1}},\quad x\in\RR^3.
\end{align}
In fact, the upper bound has already been shown in \cite[Theorem 1.2]{HiroshimaMatte2022}. Furthermore,
for every $\delta>0$, we find some $C_{\delta,A,B,m}\in(0,\infty)$ such that, with $\gamma(s)=sx$ as trial path,
\begin{align*}
\varrho_\circ(x)&\le \int_0^1\big(A^2(|sx|+1)^{2m}+2B\big)^{1/2}|x|\rd s
\leq (1+\delta)\frac{A}{m+1}|x|^{m+1}+C_{\delta,A,B,m},
\end{align*}
for all $x\in\RR^3$. In combination with Theorem~\ref{mainthmvr} this proves the lower bound in \eqref{oliver}.
\end{exam}

Spatial decay estimates can be extended by combining them with the types of localization estimates explained in the next two remarks:

\begin{rem}\label{remotherloc1}
Let $\lambda\geq E_{\rm g}$ and $\Psi(\cdot)$ be the continuous representative of some
normalized $\Psi\in\mathrm{Ran}(\one_{(-\infty,\lambda]}(H))$. In \cite[page 33]{HiroshimaMatte2022} we showed
that $\|e^{r\rd\Gamma(\om\wedge1)}\Psi(x)\|_{\cF}\leq c e^{6r\lambda}\|\Psi(x)\|_{\cF}^{1-\delta}$ for all $x\in\RR^3$,
$r>0$ and $\delta\in(0,1)$, where $c\in(0,\infty)$ solely depends on $r$, $\delta$, $g$ and $V$.
(Here we again use the notation \eqref{defhf} for differential second quantizations.)
Combining this bound with the inequality on \cite[page~62]{HiroshimaMatte2022},
we find an exponential boson number localization bound on the continuous ground state eigenvector,
\begin{align}\label{Nlocgr}
\|e^{r\rd\Gamma(1)}\gr(x)\|_{\cF}&\leq C\|\gr(x)\|_{\cF}^{1-\delta},\quad x\in\RR^3,
\end{align}
for all $r>0$ and $\delta\in(0,1)$, where $C\in(0,\infty)$ depends, besides $r$, $\delta$, $g$ and $V$, also explicitly on the
integral value in \eqref{IRcond}. These bounds can obviously be combined with \eqref{pwub}. 
\end{rem}

\begin{rem}\label{remotherloc2}
The field operator $\varphi(f)$ is interpreted as a field position observable in case
\begin{align}\label{defhreal}
f\in \mathfrak{h}_{\RR}&:=\{f\in L^2(\RR^3)|\,\bar{f}(k)=f(-k)\;\text{for a.e. $k\in\RR^3$}\}.
\end{align} 
If $f\in \mathfrak{h}_{\RR}$ satisfies $\|f\|<1/2$, then the continuous ground state eigenvector obeys the bounds
\begin{align*}
\frac{\|\gr(x)\|_{\cF}}{(1-4\|f\|^2)^{1/4}}&\leq\|e^{\varphi(f)^2}\gr(x)\|_{\cF}
\leq\frac{\|e^{\rd\Gamma(1)/(\alpha-1)}\gr(x)\|_{\cF}}{(1-4\alpha\|f\|^2)^{1/2}},\quad x\in\RR^3,
\end{align*}
for all $\alpha>1$ with $4\alpha\|f\|^2<1$, \cite[Remark~1.13 and Theorem~1.14]{HiroshimaMatte2022},
which can be combined with  \eqref{pwub}, \eqref{pwlb} and \eqref{Nlocgr}. 
Since the left hand side diverges, as $\|f\|\uparrow1/2$, the left inequality is a lower bound on the
localization of the radiation field in the ground state in an average sense. 
In fact, $\gr(x)\notin\cD(e^{\varphi(f)^2})$ for all $x\in\RR^3$, if $f\in \mathfrak{h}_{\RR}$ satisfies $\|f\|\geq1/2$,~\cite{HiroshimaMatte2022}.
\end{rem}


\section{Feynman--Kac formula and some of its implications}\label{secFK}

\noindent
In this section we mainly collect prerequisites for proving the lower bound stated in Theorem~\ref{mainthmvr}.
We start by introducing the necessary Fock space calculus in Section~\ref{ssecFock}, and present the Feynman--Kac formula
for the renormalized Nelson model together with some crucial exponential moment bounds in Section~\ref{ssecFK}.
In Section~\ref{sseccont} we give a short proof of the continuity of ground state eigenvectors and more general functions.
As a byproduct, an exponentially weighted $L^2$ to $L^\infty$ bound on the semigroup \cite{HiroshimaMatte2022,MatteMoeller2018} 
is re-derived in that subsection as well. In Section~\ref{ssecubproof}, we infer Proposition~\ref{propub} from results of \cite{HiroshimaMatte2022}.
In a crucial step in our proof of Theorem~\ref{mainthmvr} (Lemma~\ref{ag00-r}) we shall exploit the strict positivity of suitably chosen 
ground state eigenvectors in a $Q$-space representation of the Fock space. Related notation and results from 
\cite{HiroshimaMatte2022,MatteMoeller2018} are compiled in the final Section~\ref{ssecpos}.

\subsection{Some bosonic Fock space calculus}\label{ssecFock}

The bosonic Fock space over $L^2(\RR^3)$ is 
\begin{align}\label{defFock}
\cF=\CC\oplus\bigoplus_{n=1}^\infty L^2_{\rm sym}(\RR^{3n}),
\end{align}
where $L^2_{\rm sym}(\RR^{3n})\subset L^2(\RR^{3n})$ is the closed subspace comprised of all $\Phi^{(n)}\in L^2(\RR^{3n})$ such that
$\Phi^{(n)}(k_{\pi(1)},\ldots,k_{\pi(n)})=\Phi^{(n)}(k_1,\ldots,k_n)$ a.e. for all permutations $\pi$ of $\{1,\ldots,n\}$; here and below $k_j\in\RR^3$.
The vacuum vector in $\cF$ is denoted by
\begin{align}\label{vacuum}
\Omega=1\oplus0\oplus0\oplus\cdots\in \cF.
\end{align}

The selfadjoint differential second quantization of a non-negative multiplication operator $\varkappa$ in $L^2(\RR^3)$ is given by
\begin{align}\label{defhf}
(\rd\Gamma(\varkappa)\Phi)^{(n)}(k_1,\ldots,k_n)&:=\sum_{j=1}^n\varkappa(k_j)\Phi^{(n)}(k_1,\ldots,k_n),
\end{align}
and $(\rd\Gamma(\varkappa)\Phi)^{(0)}:=0$, where the domain of $\rd\Gamma(\varkappa)$ is the set of all $\Phi=\bigoplus_{n=0}^\infty\Phi^{(n)}$
for which these formulas define a new element $\rd\Gamma(\varkappa)\Phi$ of $\cF$. 
We already defined the field energy operator by $\hf=\rd\Gamma(\omega)$.
For every $x\in\RR^3$,  we further define the unitary second quantization $\Gamma(e_x)$
of the multiplication operator corresponding to the complex exponential $e_x$ introduced in \eqref{defxexk} by
\begin{align*}
(\Gamma(e_x)\Phi)^{(n)}(k_1,\ldots,k_n)&:=e_x(k_1)\dots e_{x}(k_n)\Phi^{(n)}(k_1,\ldots,k_n),
\end{align*}
and $(\Gamma(e_x)\Phi)^{(0)}:=\Phi^{(0)}$.

As usual, $\add(f)$ denotes the creation operator in $\cF$ smeared by $f\in L^2(\RR^3)$. It is given by
\begin{align}\label{defadagger}
(\add(f)\Phi)^{(n)}(k_1,\ldots,k_n)&:=\frac{1}{\sqrt{n}}\sum_{j=1}^nf(k_j)\Phi^{(n-1)}(k_1,\ldots,k_{j-1},k_{j+1},\ldots,k_n),
\end{align}
and $(\add(f)\Phi)^{(0)}:=0$.
Again, the domain of $\add(f)$ is the set of all $\Phi\in\cF$ for which the above formulas define a new element $\add(f)\Phi$ of $\cF$.
The annihilation operator and the selfadjoint field operator corresponding to $f\in L^2(\RR^3)$ are defined by
\begin{align}\label{defavp}
a(f):=\add(f)^*,\qquad \varphi(f):=(\add(f)+a(f))^{**},
\end{align}
respectively. 
If $f,g\in\cD(1/\sqrt{\om})$, then every polynomial in $\add(f)$ and $a(g)$ of order $m$ is relatively $\hf^{m/2}$-bounded, for instance
\begin{align}\label{rbvp}
\|\varphi(f)\Phi\|&\le 2\|f/\sqrt{\om\wedge 1}\|\|(1+\hf)^{1/2}\Phi\|,\quad\Phi\in\cD(\hf^{1/2}),
\end{align}
and the canonical commutation relations
\begin{align}\label{CCR}
[a(g),\add(f)]=(g,f)\one_{\cF},\quad[a(g),a(f)]=0=[\add(g),\add(f)]\quad\text{hold on $\cD(\hf)$.}
\end{align}

As shown in \cite[Lemma~17.4]{GMM2017} and \cite[\textsection1.2.7]{HL20}, the series
\begin{align*}
I_T(f)&:=\sum_{n=0}^\infty\frac{1}{n!}\add(f)^ne^{-(T/2)\hf},\quad f\in\cD(1/\sqrt{\om}),\,T>0,
\end{align*}
converge in $\cB(\cF)$, the space of bounded operators on $\cF$, and
\begin{align}\label{ITnorm}
\|I_T(f)\|&\leq \sqrt{2}e^{4\|f\|^2+8\|f/\sqrt{\om}\|^2/T}.
\end{align}
The so-obtained map $I_T$ is analytic from $\cD(1/\sqrt{\om})$, equipped with the graph norm of $1/\sqrt{\om}$, to $\cB(\cF)$.
Later on we shall use the relations
\begin{align}\nonumber
(\Omega,\Gamma(e_x)I_T(g)I_T(f)^*\Gamma(e_{-x})\Phi)_{\cF}
&=(\Gamma(e_x)I_T(f)\Omega,\Phi)_{\cF}
\\\nonumber
&=\sum_{n=0}^\infty\frac{1}{n!}(\add(e_xf)^n\Omega,\Phi)_{\cF}
\\\label{ellrel}
&=e^{-\|f\|^2/2}\sum_{n=0}^\infty\frac{1}{n!}(\varphi(e_xf)^n\Omega,\Phi)_{\cF},
\end{align}
valid for all $f,g\in\cD(1/\sqrt{\om})$, $\Phi\in\cF$ and $T>0$,
where the first two follow directly from the definitions of the involved operators and the last one from \eqref{defavp} and \eqref{CCR}.


\subsection{Feynman--Kac representation}\label{ssecFK}

The next proposition summarizes all previous results we need to know about the probabilistic analysis of the renormalized Nelson model.

In what follows we fix a filtered probability space $(\cX,\sF,(\sF_t)_{t\geq0},\PP)$ fulfilling the usual assumptions
of completeness and right-continuity of the filtration; $\EE$ denotes expectation with respect to $\PP$.
Furthermore, $B=(B_t)_{t\geq0}$ denotes a three-dimensional $(\sF_t)_{t\geq0}$-Brownian motion.
Also recall the definition of  the real Hilbert space $\mathfrak{h}_{\RR}$ in \eqref{defhreal}.
Constants called ``universal" neither depend on $g$, $\nu$, $\eta$, $V$ nor any other involved parameter such as $p$ or $T$.

\begin{prop}\label{fkf}
There exist two $\mathfrak{h}_{\RR}$-valued adapted continuous stochastic processes starting at zero, denoted $(U_{T})_{T\geq0}$ and $(\tilde{U}_{T})_{T\geq0}$,
as well as a real-valued adapted continuous stochastic process starting at zero, the complex action denoted $(S_T)_{T\geq0}$, such that the following statements hold true:
\begin{enumerate}
\item[{\rm(i)}] 
Let $U^\sharp\in\{U,\tilde{U}\}$. Then, for every $T\in(0,\infty)$, the values of the random variable 
$U_T^\sharp:\cX\to\mathfrak{h}_{\RR}$ belong $\PP$-a.s. to $\cD(1/\sqrt{\om})$.
Furthermore, there exist universal constants $c,c',C,C'\in(0,\infty)$ such that, for all $p,T\in(0,\infty)$,
\begin{align*}
\EE\big[\exp(\pm pS_T)\big]&\le Ce^{c(pg^2+p^2g^4)(1\vee T)},
\\
\EE\big[\exp\big(p\|U^\sharp_T\|^2+p\|U^{\sharp}_T/\sqrt{\om}\|^2/T\big)\big]&\le C'(1+pg^2(1\vee T))e^{c'pg^2(1+\ln(1\vee T))+cp^2g^4}.
\end{align*}
\item[{\rm(ii)}] Let $T>0$. Then 
\begin{align}\label{defJT}
J_{[0,T]}&:=e^{S_{T}}I_{T}(-U_T) I_{T}(-\tilde{U}_T)^*,
\end{align}
is a well-defined $\cB(\cF)$-valued random variable with a separable image and
\begin{align}\label{momentbdJ}
\EE\big[\|J_{[0,T]}\|_{\cB(\cF)}^p\big]&\le Ce^{c(pg^2+p^2g^4)(1\vee T)},\quad p\in(0,\infty),
\end{align}
for universal constants $c,C\in(0,\infty)$.
Finally, for every $\Psi\in \cH$, following Feynman--Kac formula holds:
\begin{align}\label{FKF}
(e^{-T\hn}\Psi)(x)&=
\EE\Big[e^{-\int_0^TV(B_s+x) \rd s} \Gamma(e_x) J_{[0,T]} \Gamma(e_{-x})\Psi(B_T+x)\Big],\quad\text{a.e. $x\in\RR^3$.}
\end{align}
\end{enumerate}
\end{prop}

\begin{proof}
The first bound in (i) with plus sign in the exponential and second bound in (i) follow directly from 
\cite[Theorem~4.9]{MatteMoeller2018} and \cite[Lemmas~3.12 and~3.17]{MatteMoeller2018}, respectively.
The first bound in (i) with minus sign is a bit unusual. However, by \cite[Remark~4.13]{MatteMoeller2018} the complex action can be split as
$S_T=u_{\Lambda,\infty,T}^{1,<}+u_{\Lambda,\infty,T}^{1,>}$ with $\Lambda\in(0,\infty)$ (here we use the notation of \cite{MatteMoeller2018} taking 
into account that the complex action is $x$-independent when only one matter particle is considered), where
$|u_{\Lambda,\infty,T}^{1,<}|\le 4\pi g^2\Lambda T$ on $\cX$ and for all $\Lambda\in(0,\infty)$ (see \cite[p.~48]{MatteMoeller2018}) and
\begin{align*}
\EE[\exp(-pu_{\Lambda,\infty,T}^{1,>})]&\leq\EE[\exp(-p\theta u_{\Lambda,\infty,T}^{1,>})]^{1/\theta}
\leq D^{1/\theta}e^{pE_{\Lambda}T+bpg^2T+b(1\vee T)/\theta},
\end{align*}
with universal constants $b,D\in[1,\infty)$ provided that we choose $\theta=1\vee(1/pg^2)$ and $\Lambda=256\pi p\theta g^2=256\pi[1\vee(pg^2)]$.
Then $D^{1/\theta}\leq D$, $b/\theta=b(1\wedge(pg^2))$ and the renormalization energy given by \eqref{defEren} satisfies
$E_\Lambda\leq b'g^2(1+0\vee\ln(pg^2))$ with another universal constant $b'\in(0,\infty)$.
When put together, these remarks complete the proof of (i). The assertions on $J_{[0,T]}$ in (ii) stem from \cite[\textsection5.2]{MatteMoeller2018}. They follow from (i),
the remarks on the maps $I_T$ in the previous subsection, in particular \eqref{ITnorm}, and from H\"{o}lder's inequality.
The Feynman--Kac formula \eqref{FKF} can be found in \cite[Theorem~5.13]{MatteMoeller2018}.
\end{proof}


 \subsection{Boundedness and continuity in the range of the semigroup}\label{sseccont}

We fix $T>0$ in the whole Section~\ref{sseccont}.
The next lemma reveals that the right hand side of \eqref{FKF} defines the unique continuous representative of $e^{-T\hn}\Psi$.
Its continuity statement is a special case of \cite[Theorem~8.8]{MatteMoeller2018} according to which $e^{-T\hn}$ maps bounded sets
in $L^p(\RR^3,\cF)$ onto equicontinuous sets of functions from $\RR^3$ to $\cF$.
We re-derive the next lemma nevertheless, because its simpler statement admits a much shorter proof. Moreover, continuity is crucial 
for obtaining lower bounds on the decay of ground state eigenvectors, 
as can be seen directly by having a glance at \eqref{defchiR}, \eqref{defalpha} and \eqref{mick-r}. 
The weighted $L^2$ to $L^\infty$ bound \eqref{L2Linfty} becomes important in Section~\ref{ssecubproof}; 
see \cite[Proposition~3.1]{HiroshimaMatte2022} for more general bounds.

\begin{lem}\label{lemcont}
Let $J:\cX\to\cB(\cF)$ be strongly measurable and such that $\|J\|^p$ is $\mathbb{P}$-integrable for some $p>3$.
Let $F:\RR^3\to\RR$ satisfy $|F(x)-F(y)|\leq L|x-y|$ for all $x,y\in\RR^3$ and some $L\in[0,\infty)$ and
let $\Psi\in \cH$ be such that also $e^F\Psi\in\cH$. Then $A\Psi:\RR^3\to\cF$ given by
\begin{align}\label{defAPsi}
A\Psi(x):=\EE\Big[e^{-\int_0^TV(B_s+x)\rd s}\Gamma(e_x)J\Gamma(e_{-x})\Psi(B_T+x)\Big],\quad x\in\RR^3,
\end{align}
is continuous and satisfies 
\begin{align}\label{L2Linfty}
e^{F(x)}\|A\Psi(x)\|_{\cF}\leq T^{-3/4}e^{6L^2T+c_{p,V_-}(1+T)}\EE[\|J\|^{p}]^{1/p}\|e^F\Psi\|_{\cH},\quad x\in\RR^3,
\end{align} 
with a constant $c_{p,V_-}\in(0,\infty)$ solely depending on $p$ and $V_-$.
\end{lem}

\begin{proof}
Let $q>3$ be such that $q^{-1}+p^{-1}+2^{-1}+6^{-1}=1$. Then H\"{o}lder's inequality and the Lipschitz estimate for $F$ imply
\begin{align}\label{HoelderL2Linfty}
e^{F(x)}\|A\Psi(x)\|_{\cF}\leq\EE[e^{6L|B_T|}]^{\frac{1}{6}}
\EE\Big[e^{q\int_0^TV_-(B_s+x)\rd s}\Big]^{\frac{1}{q}}\EE[\|J\|^p]^{\frac{1}{p}}
\EE[\|(e^F\Psi)(B_T+x)\|_\cF^2]^{\half}.
\end{align}
Here $\EE[e^{6L|B_T|}]^{1/6}\leq 2^{1/4}e^{6L^2T}$,
$\EE[\|(e^F\Psi) (B_T+x)\|_\cF^2]\leq (2\pi T)^{-3/2} \|e^F\Psi \|_\cH^2$ for all $x\in\RR^3$ and 
\begin{align*}
\sup_{x\in\RR^3}\EE\Big[e^{q\int_0^tV_-(B_s+x)\rd s}\Big]^{1/q}\le e^{tc_{p,V_-}'},\quad t\geq 0,\quad\text{for some $c_{p,V_-}'\in[0,\infty)$,}
\end{align*}
because $V_-$ belongs to the Kato class \cite{AS82}. This proves \eqref{L2Linfty}.
Now pick continuous and bounded $\Psi_n\in \cH$, $n\in\NN$, such that $\Psi_n\to\Psi$ in $\cH$ as $n\to\infty$.
Then $A\Psi_n\to A\Psi$ uniformly on $\RR^3$ by \eqref{L2Linfty}, whence it suffices to pick some $n\in\NN$
and prove continuity of $A\Psi_n$. Assume now that $V$ is continuous and compactly supported.
Since $\RR^3\ni x\mapsto\Gamma(e_x)$ is strongly continuous, the integrand in $A\Psi_n$ then is continuous in $x\in\RR^3$ point-wise on $\cX$.
The continuity of $A\Psi_n$ then follows by dominated convergence with $\mathbb{P}$-integrable majorant $e^{\|V_-\|_\infty T}\|J\|\|\Psi_n\|_\infty$.
Thus, $A\Psi$ is continuous, if $V\in C_0(\RR^3,\RR)$.
For general Kato-decomposable $V$ we may, according to \cite[Proposition~2.3 and Lemma~C.6]{BHL2000} choose 
$V_n\in C_0(\RR^3,\RR)$, $n\in\NN$, such that
\begin{align*}
c_{K,q}(n):=\sup_{x\in K}\EE\Big[\big|e^{-\int_0^TV_n(B_s+x)\rd s}-e^{-\int_0^TV(B_s+x)\rd s}\big|^q\Big]^{1/q}
\xrightarrow{\;\;n\to\infty\;\;}0,\quad q\geq 1,
\end{align*}
for all compact $K\subset\RR^3$. Define $A_n\Psi$ by \eqref{defAPsi} with $V_n$ put in place of $V$.
Then an estimation analogous to \eqref{HoelderL2Linfty} with $q>2$ chosen such that $q^{-1}+p^{-1}+2^{-1}=1$ yields
\begin{align*}
\sup_{x\in K}\|A_n\Psi(x)-A\Psi(x)\|_{\cF}\leq c_{K,q}(n)(2\pi T)^{-3/4} \EE[\|J\|^p]^{1/p}\|\Psi\|_{\cH}\xrightarrow{\;\;n\to\infty\;\;}0,
\end{align*}
for all compact $K\subset\RR^3$, where each $A_n\Psi$ is continuous. This proves continuity of $A\Psi$.
\end{proof}


\subsection{Proof of the upper bounds on the spatial decay}\label{ssecubproof}

As promised in the introduction, we next explain how our upper bounds on the spatial decay can be inferred from \cite{HiroshimaMatte2022}.

\begin{proof}[Proof of Proposition~\ref{propub}.]
We denote by $\Sigma_R$ the minimum of the spectrum of the Dirichlet realization of $H$ on $\cG_R:=\{x\in\RR^3|\,|x|>R\}$, \cite{HiroshimaMatte2022}. 
Furthermore, we define $H^0$ by putting $0$ in place of $V$ in the definition of $H$ and set $E^0:=\inf\sigma(H^0)$. 
Then $\Sigma_R\geq E^0+a(R)$ with $a(R):=\inf_{|x|\geq R} V(x)$.
We pick $\zeta\in(0,1)$ such that $1-\zeta=\sqrt{1-\eps}$ and choose $R>1$ so large that
\begin{align*}
 E^0+V-\frac{1-\eps}{2}|\nabla\varrho|^2-\frac{4}{R^2}-\lambda&\geq E^0+\eps a(R)-\frac{4}{R^2}-\lambda\geq 1\quad\text{a.e. on $\cG_R$.}
\end{align*}
Here we used \eqref{eikonaleq} to get the first inequality. Now \cite[Proposition~3.3]{HiroshimaMatte2022} directly implies
\begin{align*}
\|e^{(1-\eps)\varrho}\one_{(-\infty,\lambda]}(H)\|&\leq\frac{c}{\zeta^2}(\Sigma_R-E_{{\rm g}}+1)^2\Big(1+\underset{|x|\leq2R}{{\rm ess\,sup}}|\nabla\varrho(x)|^2\Big)
\exp\Big(\max\limits_{|x|\leq2R}\varrho(x)\Big),
\end{align*}
for some universal constant $c\in(0,\infty)$.

To turn this into a point-wise bound we argue similarly as in \cite[p.~34]{HiroshimaMatte2022}: Since $\varrho$ is not necessarily globally
Lipschitz continuous, we fix $y\in\RR^3$ and pick some 
$f\in C^\infty(\RR,\RR)$ such that $f(r)=r$ for all $r\in(-\infty,\varrho(y)]$, $f(r)=\varrho(y)+1$ for all $r\in[\varrho(y)+2,\infty)$
and $0\le f' \le 1$. Then $\theta:=f\circ\varrho$ is globally Lipschitz continuous and \eqref{eikonaleq} entails
$|\nabla\theta|=f'(\varrho)|\nabla\varrho|=f'(\varrho)\sqrt{2V}\le\sup_{\varrho(z)\leq\varrho(y)+2}\sqrt{2V(z)}$, a.e.
By \eqref{lbvarrho} the lower bound $\varrho(z)\geq\sqrt{2}|z|\geq|z|$ holds for all $z\in\RR^3$, which together with \eqref{extraV}
implies $|\nabla\theta|\leq L$, a.e., with $L:=\sqrt{2D_\delta} e^{\delta(\varrho(y)+2)/2}$. 
Here we fix $\delta>0$ such that $3\delta/2=\eps$. Writing
\begin{align*}
\Psi&=e^{-TH}\one_{(-\infty,\lambda]}(H)\{e^{TH}\one_{(-\infty,\lambda]}(H)\}\Psi,
\end{align*}
we first observe that the continuity statement in Proposition~\ref{propub} follows from \eqref{FKF} and Lemma~\ref{lemcont}.
Using $\varrho(y)=\theta(y)$ and $\theta\leq\varrho$, we further deduce that
\begin{align*}
e^{(1-\eps/2)\varrho(y)}\|\Psi(y)\|_{\cF}&=\|(e^{(1-\eps/2)\theta}\Psi)(y)\|_{\cF}
\\
&\le \|e^{(1-\eps/2)\theta}e^{-TH}e^{-(1-\eps/2)\theta}\|_{2,\infty}
\|e^{(1-\eps/2)\varrho}\one_{(-\infty,\lambda]}(H)\|e^{T\lambda}.
\end{align*}
Here $\|\cdot\|_{2,\infty}$ is the norm on the space of bounded operators from $\cH$ to $L^\infty(\RR^3,\sF)$.
We choose $T:=1/L^2\leq1$. Employing \eqref{momentbdJ}, \eqref{FKF} and \eqref{L2Linfty} (with $V_-=0$, $J=J_{[0,T]}$ and, say, $p=4$), 
we then find $e^{T\lambda}\leq e^{0\vee\lambda}$ and a solely $g$-dependent constant $c_{g}\in(0,\infty)$ such that
\begin{align*}
 \|e^{(1-\eps)\theta}e^{-TH}e^{-(1-\eps)\theta}\|_{2,\infty}
 &\leq c_ge^{3\delta/2}D_\delta^{3/4} e^{3\delta\varrho(y)/4}=ce^{\eps}D_{2\eps/3}^{3/4}e^{\eps\varrho(y)/2}.
\end{align*}
Since $y\in\RR^3$ was arbitrary, these remarks prove \eqref{pwub}.
\end{proof}


\subsection{Positivity of ground state eigenvectors}\label{ssecpos}

Let $L^2(Q,\mu)$ be a $Q$-space representation of $\cF$. By definition, this means that
$(Q, \Sigma,\mu)$ is a probability space and there exists an isomorphism $\sI:\cF\to L^2(Q,\mu)$ having the following properties:
Firstly, $\sI\Omega$ is constantly equal to $1$ on $Q$. Secondly,
every unitarily transformed field operator $\sI \varphi(f)\sI^*$ with $f\in\mathfrak{h}_\RR$ is the maximal
operator of multiplication with a centered Gaussian random variable $\phi(f):Q\to\RR$ having variance $\|f\|^2$.
That is, $\EE_\mu[\exp(z\phi(f))]=\exp(z^2\|f\|^2/2)$ for all $z\in\CC$, where $\EE_\mu$ denotes expectation with respect to $\mu$.
Finally, $\Sigma$ is generated by all $\phi(f)$ with $f\in\mathfrak{h}_{\RR}$.
We refer to \cite{HL20} for constructions of $(Q, \Sigma,\mu)$ and $\sI$.

It has been shown in \cite{Matte2016,MatteMoeller2018} that all operators of the form
$\sI\Gamma(e_x)I_T(f)I_T(g)^*\Gamma(e_{-x})\sI^*$ with $T>0$, $x\in\RR^3$ and $f,g\in\mathfrak{h}_{\RR}$ are positivity improving.
In conjunction with the Feynman--Kac formula this implies that $\sI e^{-TH}\sI^*$ with $T>0$ improves positivity
on $L^2(\RR^3,L^2(Q,\mu))$; see \cite{HiroshimaMatte2022} and \cite{MatteMoeller2018}. (Here we apply $\sI$ point-wise for every $x\in\RR^3$
and $\Psi\in L^2(\RR^3,L^2(Q,\mu))$ is strictly positive, iff $\Psi(x)>0$ in $L^2(Q,\mu)$ for a.e. $x\in\RR^3$.)
According to the Perron--Frobenius--Faris theorem, this implies under the condition \eqref{IRcond} 
that the ground state eigenvector $\gr$ of $H$ can be chosen such that $\sI\gr$ is strictly positive. For the unique continuous representative
$\gr(\cdot)$ of this particular choice of $\gr$, we find
\begin{align}\label{ellpos}
\ell(x)&:=(\Omega,\gr(x))_\cF=\EE_{\mu}[\sI\gr(x)]>0,\quad x\in \RR^3,
\\\label{defchiR}
\chi(R)&:=\inf_{|z|\leq R}\ell(z)>0,\quad R>0.
\end{align}


\section{Point-wise lower bound on the spatial decay}\label{secproof}

\noindent
In this section we prove Theorem~\ref{mainthmvr}.
The main steps of the proof are presented in Section~\ref{ssecgenlowerbound}, where
arguments from \cite{CS81} are combined with
new bounds on contributions coming from the radiation field and the exponential moment bounds of Proposition~\ref{fkf}(i).
The final result of Section~\ref{ssecgenlowerbound} is a general lower bound on $\ell(x)=(\Omega,\gr(x))_\cF$
involving a Lipschitz path from $x$ to the origin. Optimizing over such paths as in \cite{CS81}, 
we then encounter the Agmon distance, as explained in Section~\ref{s4}.

\subsection{A general lower bound on the spatial decay}\label{ssecgenlowerbound}

Throughout this section we assume that the infrared regularity condition \eqref{IRcond} is fulfilled and we only
consider continuous potentials $V$ satisfying $V(x)\to\infty$ as $|x|\to\infty$, so that $E_{\rm g}$ is an eigenvalue.
We choose the ground state eigenvector $\gr$ of $H$ such that $\sI\gr$ is strictly positive and
denote by $\gr(\cdot):\RR^3\to\cF$ its continuous representative.
Then \eqref{ellpos} holds, and to obtain \eqref{pwlb} it suffices to bound $\ell(x)$ from below. To this end we exploit that
\begin{align}\label{eva1}
\ell(x)&=\EE\Big[e^{-\int_0^T(V(B_s+x)-\eg)\rd s} (\Omega,\Gamma(e_{x})J_{[0,T]}\Gamma(e_{-x})\gr(B_T+x))_{\cF}\Big],
\end{align}
for all $x\in\RR^3$ by Proposition~\ref{fkf} and the relation $\gr=e^{-T\hn+T\eg}\gr$, where furthermore
\begin{align}\label{eva2}
(\Omega,\Gamma(e_{x})J_{[0,T]}\Gamma(e_{-x})\gr(y))_{\cF}
&=e^{S_{T}-\frac{1}{2}\|\tilde{U}_T\|^2}(e^{-\varphi(e_x\tilde{U}_T)}\Omega,\gr(y)),\quad y\in\RR^3,
\end{align}
on $\cX$ by \eqref{ellrel}, \eqref{defJT} and the spectral calculus for $\varphi(f)$.
The appearance of the field operator on the right hand side of \eqref{eva2} permits to apply Jensen's inequality in $Q$-space:

\begin{lem}\label{ag00-r}
Let $y\in\RR^3$ and $f\in \mathfrak{h}_{\RR}$. Then
\begin{align}\label{eva3}
(e^{-\varphi(f)}\Omega, \gr(y))_\cF&\geq\ell(y)\exp\bigg(-\frac{\|f\|\|\gr\|_\infty}{\ell(y)}\bigg).
\end{align}
\end{lem}

\begin{proof}
Let $\sI:\cF\to L^2(Q,\mu)$ be the isomorphism introduced in Section~\ref{ssecpos}, so that
$\sI \varphi(f)\sI^*$ becomes multiplication with the Gaussian random variable $\phi(f)$, $\sI\Omega=1$ and $\sI\gr(y)>0$, $\mu$-almost surely.
By \eqref{ellpos} and Jensen's inequality we then find
\begin{align*}
\frac{(e^{-\phi(f)} , \sI\gr(y))_{L^2(Q,\mu)}}{\ell(y)}
\geq
\exp\bigg(-\frac{(\phi(f) , \sI\gr(y))_{L^2(Q,\mu)}}{\ell(y)}\bigg).
\end{align*}
Here  $(\phi(f) , \sI\gr(y))_{L^2(Q,\mu)}=(\add(f)\Omega,\gr(y))_{\cF}\le \|f\|\|\gr\|_\infty$ holds by \eqref{defadagger}, \eqref{defavp} 
and the relation $a(f)\Omega=0$.
\end{proof}

Recall the notation \eqref{defVplus} and \eqref{defchiR} before reading the next lemma.

\begin{lem}\label{ag1-r}
Denote by $\tau>0$ the smallest eigenvalue of $-1/2$ times the Dirichlet--Laplacian on the unit ball in $\RR^3$. Abbreviate 
\begin{align}\label{defalpha}
\alpha:=\chi(1) \exp\bigg(-\frac{\|\gr\|_\infty^2}{2\chi(1)^2}\bigg).
\end{align} 
Then there exist universal constants $c,C\in(0,\infty)$
such that, for all $T>0$, $x\in\RR^3$ and $q\in {\rm Lip}([0,T],\RR^3)$ with $q(0)=x$ and $q(T)=0$ and for all $p\in(1,\infty)$,
\begin{align}\label{mick-r}
\ell(x)&\geq\frac{\alpha}{C^{p-1}}\cdot
e^{-\frac{p}{2}\int_0^T|\dot q(s)|^2\rd s-\int_0^TV_{\circ}(q(s))\rd s+T(\eg-p\tau)-c(g^2+\frac{g^4}{p-1})(1\vee T)}.
\end{align}
\end{lem}

\begin{proof}
We fix $x\in\RR^3$ and $T>0$ and pick $q\in {\rm Lip}([0,T],\RR^3)$ such that $q(0)=x$ and $q(T)=0$. Introducing the events
\begin{align*}
M_T:=\Big\{\sup_{s\in[0,T]}|B_s+x-q(s)|\leq 1\Big\},\quad Q_T:=\Big\{\sup_{s\in[0,T]}|B_s|\leq 1\Big\},
\end{align*}
and combining \eqref{eva1} through \eqref{eva3}, we find
\begin{align*}
\ell(x)&\geq
\EE\bigg[\one_{M_T}e^{-\int_0^T(V(B_s+x)-\eg)\rd s}
e^{S_{T}-\frac{1}{2}\|\tilde{U}_{T}\|^2}\exp\Big(-\frac{\|\tilde{U}_{T}\|\|\gr\|_\infty}{\ell(B_T+x)}\Big)\ell(B_T+x)\bigg].
\end{align*}
Note that $|B_T+x|\leq1$ and, hence, $\ell(B_T+x)\geq \chi(1)$ on $M_T$. Thus,
\begin{align*}
\ell(x)&\geq\alpha\cdot\EE\bigg[\one_{M_T}e^{-\int_0^T(V(B_s+x)-\eg)\rd s}
e^{S_{T}-\|\tilde{U}_{T}\|^2}\bigg].
\end{align*}
Let $p\in(1,\infty)$ and $p'$ be the exponent conjugate to $p$. Then H\"{o}lder's inequality yields
\begin{align}\label{adam1}
\ell(x)&\geq \alpha\cdot \frac{\EE\Big[\one_{M_T}e^{-\int_0^T(V(B_s+x)-\eg)\rd s/p}\Big]^{p}}{
\EE\Big[e^{(p'/p)\|\tilde{U}_{T}\|^2-(p'/p)S_{T}}\Big]^{p/p'}}.
\end{align}
Since $p/p'=p-1$, the denominator in \eqref{adam1} can be bounded from above by
\begin{align*}
\EE\Big[e^{(p'/p)\|\tilde{U}_{T}\|^2-(p'/p)S_{T}}\Big]^{p/p'}
&=\EE\Big[e^{\|\tilde{U}_{T}\|^2/(p-1)-S_{T}/(p-1)}\Big]^{p-1}
\leq C^{p-1}e^{c(g^2+cg^4/(p-1))(1\vee T)},
\end{align*}
where the inequality follows from the exponential moment bounds of Proposition~\ref{fkf}(i).
Using that $V(B_s+x)\le V_{\circ}(q(s))$ on $M_T$, we estimate the numerator in \eqref{adam1} from below as
\begin{align*}
\EE\Big[\one_{M_T}e^{-\int_0^T(V(B_s+x)-\eg)\rd s/p}\Big]^{p}
&\geq e^{-\int_0^TV_{\circ}(q(s))\rd s+T\eg}\PP[M_T]^p
\end{align*}
To conclude, we recall the bounds
\begin{align}\label{CSbound}
\PP[M_T]&\geq e^{-\half\int_0^T|\dot{q}(s)|^2\rd s}\PP[Q_T],\quad\PP[Q_T]\geq e^{-\tau T},
\end{align}
from Lemma~2.2 and~2.3, respectively, in  \cite{CS81}.
\end{proof}

For completeness we recall from \cite{CS81} how the first inequality in \eqref{CSbound} is obtained: 
Consider the martingale $m=(m_t)_{t\in[0,T]}$ given by $m_t:=-\int_0^t\dot{q}(s)\rd B_s$, $t\in[0,T]$, 
which has the deterministic quadratic variation $\langle m\rangle_t=\int_0^t|\dot{q}(s)|^2\rd s$, $t\in[0,T]$. 
As a consequence of Girsanov's theorem or the Cameron--Martin theorem,
the law of the process $(B_t)_{t\in[0,T]}$ under $\PP$ is the same as the law of 
$(B_t+\int_0^t\dot{q}(s)\rd s)_{t\in[0,T]}=(B_t+q(t)-x)_{t\in[0,T]}$
under the probability measure $\xi_T\odot\PP$ whose density $\xi_T:=\exp(m_T-\langle m\rangle_T/2)$ is the stochastic  
exponential of $m$ at $T$. Thus, 
\begin{align*}
\PP[M_T]&=\EE[\xi_T\one_{Q_T}]
\ge \PP[Q_T]\exp\bigg(\frac{\EE[(m_T-\langle m\rangle_T/2)\one_{Q_T}]}{\PP[Q_T]}\bigg)
=e^{-\langle m\rangle_T/2}\PP[Q_T],
\end{align*}
where we applied Jensen's inequality in the second step. In the last one we used $\EE[m_T\one_{Q_T}]=0$, which holds
since $B$ has the same law under $\PP$ as $-B$.


\subsection{Optimal paths and geodesic distance}\label{s4}

Let $x\in\RR^3\setminus\{0\}$. Lemma~\ref{ag1-r}, where $p>1$ is arbitrary, leads to the minimization of the classical action functional
\begin{align}\label{defAV}
&\sA_V(q):=\int_0^T\lk \half |\dot q(t)|^2+V(q(t))\rk \rd t,
\end{align}
with $V_{\circ}$ is put in place of $V$ and the minimization running over all $T>0$ and all paths
\begin{align}\label{defcCast}
q\in\cC^\ast_T(x):=\{q\in {\rm Lip}([0,T],\RR^3)\mid q(0)=x, q(T)=0\}.
\end{align}
It is well-known that $\varrho(x)=\tilde{\varrho}(x)$ with $\varrho$ given by \eqref{defcCT} through \eqref{defvrx} and
\begin{align}\label{miniA}
\tilde{\varrho}(x)&:=\inf\{\sA_V(q)|\,q\in\cC^\ast_T(x),\,T>0\}.
\end{align}
Moreover, there exist $T(x)>0$ and $q_x\in\cC^\ast_{T(x)}(x)$ minimizing the variational problem \eqref{miniA}, 
so that $\varrho(x)=\sA_V(q_x)$, and such that $T(x)$ and $\varrho(x)$ are related by the inequality in Lemma~\ref{lemCSTvr}; 
see \cite[\textsection2]{CS81} for proofs or Appendix~\ref{appdist} and in particular Lemma~\ref{lemJac}.

\begin{proof}[Proof of Theorem~\ref{mainthmvr}.]
Let $\eps>0$ and abbreviate $c_{*}:=\max\{0,(1+\eps/2)\tau+c(g^2+2g^4/\eps)-\eg\}$. Put $a_\circ(R):=\min_{|z|\geq R}V_{\circ}(z)$, $R>0$, and
choose $R_0>0$ so large that $c_{*}/(2a_\circ(R_0))<\eps/4$. After that choose $R_1>R_0$ so large that
$c_{*}\max_{|z|=R_0}\varrho_{\circ}(z)/(\sqrt{8}R_1)<\eps/4$. Pick $x\in\RR^3$ with $|x|\geq R_1$.
Also $V_{\circ}$ satisfies \eqref{hypV}, whence all the above remarks on the variational problems \eqref{defvrx} and \eqref{miniA}
apply to $V_{\circ}$ as well. Hence, we find $T_{\circ}(x)>0$ and $q_x^\circ\in\cC^\ast_{T_\circ(x)}(x)$ such that $\varrho_\circ(x)=\sA_{V_\circ}(q_x^\circ)$
and $T_\circ(x)$ and $\varrho_\circ(x)$ are related by the inequality in Lemma~\ref{lemCSTvr} applied to $V_\circ$.
By the above choices of $R_0$ and $R_1$, the latter inequality implies $c_{*}T_\circ(x)<\eps\varrho_\circ(x)/2$.
Plugging $p:=1+\eps/2$ and the minimizing pair $(q_x^\circ,T_\circ(x))$ into \eqref{mick-r}, we thus find
\begin{align*}
\ell(x)\geq C_{*}\exp\big(-(1+\eps/2)\varrho_{\circ}(x)-c_{*}T_\circ(x)\big)\geq C_{*}\exp(-(1+\eps)\varrho_\circ(x)),
\end{align*}
with $C_{*}:=\alpha e^{-c(g^2+2g^4/\eps)}/C^{\eps/2}$.
If $|x|<R_1$, then $\ell(x)\geq \chi(R_1)\exp(-(1+\eps)\varrho_\circ(x))$ with $\chi(R_1)$ given by \eqref{defchiR}.
\end{proof}


\appendix

\section{Basic properties of the geodesic distance}\label{appdist}

\noindent
In this appendix we collect well-known facts related to the Agmon distance which have been used in the main text,
and give references for all results whose proofs are not repeated here.

In what follows we merely assume that $V:\RR^3\to\RR$ is continuous with $V\geq1$, and
$X\subset\RR^3$ is convex and non-empty. Recalling \eqref{defcCT} we introduce the path space
\begin{align*}
\cC(x,y;X):=\{\gamma\in\sC_1(x,y)|\,\gamma([0,1])\subset \overline{X}\},\quad x,y\in\RR^3,
\end{align*}
as well as an abbreviation for the Lagrangian,
\begin{align*}
F(x,v):=\sqrt{2V(x)}|v|,\quad x,v\in\RR^3.
\end{align*}
For all $x,y\in X$, we finally generalize \eqref{defvrx} by
\begin{align}\label{defvrxgen}
d_X(x,y):=\inf_{\gamma\in\cC(x,y;X)}\sL_V(\gamma)=\inf_{\gamma\in\cC(x,y;X)}\int_0^1F(\gamma(s),\dot\gamma(s))\rd s.
\end{align}

\begin{lem}\label{lemdX}
Let $x,y\in X$, $x\not=y$. 
Then the parametric variational problem \eqref{defvrxgen} has a (not necessarily unique) 
minimizer $\gamma_{x,y;X}\in\cC(x,y;X)$ which can be chosen such that $F(\gamma_{x,y;X},\dot{\gamma}_{x,y;X})=d_X(x,y)>0$ a.e. on $[0,1]$.
\end{lem}

\begin{proof}
Since $F$ is continuous on $\RR^6$,
$F(x,v)\ge\sqrt{2}|v|$ for all $x,v\in\RR^3$ and $F(x,\cdot)$ is convex and smooth on $\RR^3\setminus\{0\}$ for every $x\in\RR^3$, 
all assertions follow directly from, e.g., \cite[Theorem~1 and Remark~2 in \textsection8.4.2]{GiaquintaHildebrandt2}.
The proofs given there are based on invariance under re-parametrizations and lower semi-continuity arguments.
\end{proof}

Since two minimizers $\gamma_{x,z;X}$ and $\gamma_{z,y;X}$ can be combined and re-parametrized, we clearly have
the triangle inequality $d_X(x,y)\leq d_X(x,z)+d_X(z,y)$ for all $x,y,z\in X$. For all $R>0$ and $x,y,z\in X$
satisfying $|x|,|y|\leq R$, this implies the local Lipschitz bound
\begin{align}\label{dXLip}
|d_X(x,z)-d_X(y,z)|&\leq d_X(x,y)\leq\int_0^1\sqrt{2V(sx+(1-s)y)}|x-y|\rd s\le L_R|x-y|,
\end{align}
with $L_R:=\max_{|z|\leq R}\sqrt{2V(z)}$. Another simple bound that will be useful in a moment is
\begin{align}\label{lbvarrho}
d_X(x,y)\geq \sqrt{2}\int_0^t|\dot{\gamma}_{x,y;X}(s)|\rd s\geq\sqrt{2}|\gamma_{x,y;X}(t)-y|,\quad t\in[0,1],
\end{align}
for all $x,y\in X$, $x\not=y$, where $V\geq1$ was used.

Our assumptions on $V$ are translation invariant, whence we shall restrict our attention to the
distance to the origin $\varrho$ given by \eqref{defvrx} in what follows.

\begin{lem}\label{lemeikonal}
$\varrho$ is locally Lipschitz continuous on $\RR^3$ and satisfies \eqref{eikonaleq}.
\end{lem}

\begin{proof}
Let $B_r\subset\RR^3$ be the open ball of radius $r>0$ about $0$. Fix $R\in\NN$.
If $x\in B_R\setminus\{0\}$ and $M:=\max_{|z|\leq R}\varrho(z)\geq \sqrt{2}R$, then \eqref{lbvarrho} shows 
that the image of any minimizer $\gamma_{x,0;\RR^3}$ is contained in $B_{M}$. 
This implies $\varrho=\varrho_{M}:=d_{B_M}(\cdot,0)$ on $B_R$. For the variational problems defining $\varrho_{M}$ involve
smaller sets of paths than those defining $\varrho$, so that $\varrho\leq\varrho_{M}$ on $B_M$.
It follows directly from \cite[Theorem~5.2]{Lions1982} that $\varrho_{M}$ is a viscosity solution of the eikonal equation on $B_M\setminus\{0\}$.
Since $\varrho_{M}$ is Lipschitz continuous by \eqref{dXLip}, this
shows \cite[Remark~1.11]{Lions1982} that $|\nabla\varrho_M|^2=2V$ a.e. on $B_M\setminus\{0\}$.
Since $R\in\NN$ was arbitrary, this implies \eqref{eikonaleq}.
\end{proof}

Let $x\in\RR^3\setminus\{0\}$. We shall next explain why Jacobi's geometric least action principle works in our situation,
i.e., why minimizers of the variational problem for $\varrho(x)$ yield upon re-parametrization minimizers
of the variational problem for $\tilde{\varrho}(x)$ given by \eqref{defAV} through \eqref{miniA}. We follow the discussion in \cite{CS81}
which exploits that the Lagrange function in \eqref{defAV} depends quadratically on the velocities $\dot{q}(t)$; 
see, e.g., \cite[\textsection5.3]{Lions1982} for generalizations.

Given any $T>0$ and $q\in \cC^\ast_T(x)$, the re-parametrization $\gamma_{q}(s):=q(T(1-s))$, $s\in[0,1]$, defines a path $\gamma_{q}\in\cC_1(x,0)$.
Now, for $a,b\geq0$, we have $\sqrt{2}ab\leq a^2+b^2/2$ with equality if and only if $\sqrt{2}a=b$. This and the invariance of
$\sL_V$ under re-parametrizations implies
\begin{align*}
\sL_V(\gamma_{q})=\int_0^T\sqrt{2V(q(t))}|\dot{q}(t)|\rd t\leq\sA_V(q),
\end{align*}
with equality if and only if 
\begin{align}\label{iffLA}
\sqrt{2V(q(t))}=|\dot{q}(t)|,\quad \text{a.e. $t\in[0,T]$.}
\end{align}
Thus, $\varrho(x)\leq \tilde{\varrho}(x)$. 

By virtue of Lemma~\ref{lemdX}, we next pick $\gamma_x:=\gamma_{x,0;\RR^3}\in\cC_1(x,0)$ satisfying $F(\gamma_x,\dot{\gamma}_x)=\varrho(x)$, a.e.,
and in particular $\varrho(x)=\sL_V(\gamma_x)$. Setting
\begin{align}\label{defTx}
\tau(s):=\int_0^s\frac{\varrho(x)}{2V(\gamma_x(r))}\rd r,\quad s\in[0,1],\quad T(x):=\tau(1),
\end{align}
we see that $\tau$ is $C^1$ with $\inf\tau'>0$, whence it has an inverse function $\sigma:=\tau^{-1}:[0,T(x)]\to[0,1]$
with the same properties. Defining $q_x(t):=\gamma_x(\sigma(T(x)-t))$, $t\in[0,T(x)]$, and using that $F(\gamma_x,\dot{\gamma}_x)=\varrho(x)$, a.e., 
we find that $q_x\in\cC_{T(x)}^*(x)$ satisfies \eqref{iffLA}. Therefore,
\begin{align*}
\sA_V(q_x)&=\sL_V(\gamma_{q_x})=\int_0^1\sqrt{2V(\gamma_x(\sigma(T(x)s)))}|\dot{\gamma}_x(\sigma(T(x)s))|\cdot\frac{\rd}{\rd s}\sigma(T(x)s)\rd s=\sL_V(\gamma_x).
\end{align*}
We have re-proved the following result:

\begin{lem}\label{lemJac}
Let $x\in\RR^3\setminus\{0\}$. Then $\varrho(x)=\tilde{\varrho}(x)$ and the pair $T(x)>0$ and $q_x\in\cC^\ast_{T(x)}(x)$ 
constructed in the preceding paragraph minimizes the variational problem \eqref{miniA}.
\end{lem}

The inequality stated in the next lemma can be read off from a proof in \cite{CS81} that we essentially copy here.

\begin{lem}\label{lemCSTvr}
Let $R_1> R_0>0$ and set $a(R_0):=\inf_{|z|\geq R_0}V(z)$.
Pick $x\in\RR^3$ with $|x|\geq R_1$ and let $T(x)$ be given by \eqref{defTx} 
for some $\gamma_x\in\cC_1(x,0)$ satisfying $F(\gamma_x,\dot{\gamma}_x)=\varrho(x)$, a.e. Then
\begin{align*}
T(x)&\le \bigg(\frac{1}{2a(R_0)}+\frac{1}{\sqrt{8}R_1}\cdot\max_{|z|=R_0}\varrho(z)\bigg)\varrho(x).
\end{align*}
\end{lem}

\begin{proof} 
We write $\sA:=\sA_V$ for short. Set $t_0:=\min\{t>0|\,|q_x(t)|=R_0\}$ and $y:=q_x(t_0)$, 
where $q_x$ is the minimizer constructed by means of $\gamma_x$ prior to Lemma~\ref{lemJac}. 
Then $\eta(t):=q_x(t+t_0)$, $t\in[0,T(x)-t_0]$, defines a minimizer for $\inf\{\sA(q)|\,q\in\cC_T^*(y),\,T>0\}$. 
For, if there were $T_0>0$ and $\eta_0\in\cC_{T_0}^*(y)$ such that $\sA(\eta_0)<\sA(\eta)$, then
we could define $q_*\in \cC_{T_0+t_0}^*(x)$ setting $q_*(t):=q_x(t)$ for $t\in[0,t_0]$ and $q_*(t):=\eta_0(t-t_0)$ for $t\in[t_0,T_0+t_0]$
and obtained the contradiction $\sA(q_*)<\sA(q_x)$. Since $q_x$ satisfies \eqref{iffLA}, we deduce that
\begin{align*}
\varrho(y)&=\sA(\eta)=\int_{t_0}^{T(x)}2V(q_x(t))\rd t\geq2(T(x)-t_0),
\\
\varrho(x)&=\sA(q_x)\geq\int_0^{t_0}2V(q_x(t))\rd t\geq 2a(R_0)t_0.
\end{align*}
Combined this yields
\begin{align*}
\frac{T(x)}{\varrho(x)}&\leq \frac{t_0}{\varrho(x)}+\frac{1}{2\varrho(x)}\cdot\varrho(y)\leq \frac{1}{2a(R_0)}+\frac{1}{\sqrt{8}|x|}\cdot\max_{|z|=R_0}\varrho(z),
\end{align*}
where we also took \eqref{lbvarrho} with $\gamma_{x,0;\RR^3}=\gamma_x$ and $\gamma_{x}(1)=x$ into account.
\end{proof}


\bigskip

\noindent
{\bf Acknowledgements:}
FH is grateful for the kind hospitality at Aalborg University in May 2022 and at Polytechnic University of Milano in April 2022. 
This work has been partially done there. 
This work is also financially supported by JSPS KAKENHI 16H03942, JSPS KAKENHI 20H01808,
JSPS KAKENHI 20K20886 and JSPS KAKENHI 23K20217.
OM thanks the Research Institute for Mathematical Sciences at Kyoto University for their support and generous hospitality
during the workshop ``Mathematical Aspects of Quantum Fields and Related Topics'' in November 2024.


\begin{thebibliography}{42}

\bibitem{Agmon1982}
S.~Agmon.
\newblock {\em {Lectures on exponential decay of solutions of second-order
 elliptic equations: bounds on eigenfunctions of $N$-body Schr\"odinger operators}}.
\newblock Mathematical Notes, vol.~29, Princeton University Press, Princeton, New Jersey (1982)

\bibitem{AS82}
M.~Aizenman and B.~Simon.
\newblock {Brownian motion and Harnack's inequality for Schr\"odinger operators}.
\newblock {\em Comm. Pure Appl. Math.} \textbf{35}, 209--271 (1982)

\bibitem{Ammari2000}
Z.~Ammari.
\newblock {Asymptotic completeness for a renormalized non-relativistic
  Hamiltonian in quantum field theory: the Nelson model}.
\newblock {\em Math. Phys. Anal. Geom.} \textbf{3}, 217--285 (2000)

\bibitem{BFS1998b}
V.~Bach, J.~Fr{\"o}hlich and I.~M.~Sigal.
\newblock Quantum electrodynamics of confined nonrelativistic particles.
\newblock {\em Adv. Math.} \textbf{137}, 299--395 (1998)

\bibitem{BHLMS01}
V.~Betz, F.~Hiroshima, J.~L{\"o}rinczi, R.~A. Minlos, and H.~Spohn.
\newblock {Ground state properties of the Nelson Hamiltonian -- A Gibbs
  measure-based approach}.
\newblock {\em Rev. Math. Phys.} \textbf{14}, 173--198 (2002)


\bibitem{BHL2000}
K.~Broderix, D.~Hundertmark and H.~Leschke.
\newblock Continuity properties of Schr\"odinger semigroups with magnetic fields.
\newblock {\em Rev. Math. Phys.} \textbf{12}, 181--225 (2000) 

\bibitem{car78}
R.~Carmona.
\newblock {Pointwise bounds for Schr\"odinger eigenstates}.
\newblock {\em Commun. Math. Phys.} \textbf{62}, 97--106 (1978)

\bibitem{CMS90}
R.~Carmona, W.C.~Masters, and B.~Simon.
\newblock {Relativistic Schr\"odinger operators: Asymptotic behavior of the eigenfunctions}.
\newblock {\em J. Funct. Anal.} \textbf{91}, 117--142 (1990)

\bibitem{CS81}
R.~Carmona and B.~Simon.
\newblock {Pointwise bounds on eigenfunctions and wave packets in $N$-body quantum system}.
\newblock {\em Commun. Math. Phys.} \textbf{80}, 59--98 (1981)

\bibitem{ger00}
C.~G{\'e}rard.
\newblock {On the existence of ground states for massless Pauli-Fierz Hamiltonians}.
\newblock {\em Ann. Henri Poincar{\'e}} \textbf{1}, 443--459 (2000)

\bibitem{GiaquintaHildebrandt2}
M.~Giaquinta and S.~Hildebrandt.
\newblock {\em Calculus of variations II.}
\newblock Grundlehren der mathematischen Wissenschaften, vol.~311, Springer, Berlin-Heidelberg (1996)

\bibitem{Griesemer2004}
M.~Griesemer. {Exponential decay and ionization thresholds in non-relativistic quantum electrodynamics.}
{\em J. Funct. Anal.} \textbf{210}, 321--340 (2004) 

\bibitem{GriesemerWuensch2018}
M.~Griesemer and A.~W\"{u}nsch. On the domain of the Nelson Hamiltonian.
{\em J. Math. Phys.} \textbf{59}, 042111, 21 pages (2018) 

\bibitem{GriesemerKussmaul2024}
M.~Griesemer and V.~Ku{\ss}maul. Pointwise bounds on eigenstates in non-relativistic QED.
arXiv:2307.14986, 15 pages (2024, preprint)

\bibitem{GHL14}
M.~Gubinelli, F.~Hiroshima, and J.~L{\"o}rinczi.
\newblock {Ultraviolet renormalization of the Nelson Hamiltonian through
  functional integration}.
\newblock {\em J. Funct. Anal.} \textbf{267}, 3125--3153 (2014)

\bibitem{GMM2017}
B.~G\"{u}neysu, O.~Matte and J.~S.~M{\o}ller.
Stochastic differential equations for models of non-relativistic matter interacting
with quantized radiation fields,
{\em Probab. Theory Relat. Fields} \textbf{167}, 817--915 (2017)

\bibitem{HiHi2010}
T.~Hidaka and F.~Hiroshima. Pauli--Fierz model with Kato-class potentials and exponential decays. 
{\em Rev. Math. Phys.} \textbf{22}, 1181--1208 (2010)

\bibitem{hir14}
F.~Hiroshima.
\newblock {Functional integral approach to semi-relativistic Pauli--Fierz models}.
\newblock {\em Adv. Math.} \textbf{259}, 784--840 (2014)

\bibitem{Hiroshima2019}
F.~Hiroshima. Pointwise exponential decay of bound states of the Nelson model with Kato-class potentials. 
In: T.M.~Rassias and V.A.~Zagrebnov (Editors). {\em Analysis and operator theory}. 
Springer Optimization and Its Applications, vol~146, pages 225--250. Springer, Cham (2019)

\bibitem{HL20}
F.~Hiroshima and J.~L{\H o}rinczi.
\newblock {\em Feynman--Kac type theorems and its applications, volume 2}. Second edition,
\newblock De Gruyter, Berlin (2020)

\bibitem{HiroshimaMatte2022}
F.~Hiroshima and O.~Matte.
\newblock {Ground states and their associated Gibbs measures in the renormalized Nelson model}.
\newblock {\em Rev. Math.Phys.} \textbf{34}, 2250002, 84 pages (2022)

\bibitem{HS96}
P.D.~Hislop and I.M. Sigal.
\newblock {\em Introduction to Spectral Theory.}
\newblock Applied Mathematical Sciences, vol. 113. Springer, New York (1996)

\bibitem{KoenenbergMatte2013}
M.~K\"{o}nenberg and O.~Matte. 
On enhanced binding and related effects in the non- and semi-relativistic Pauli--Fierz models.
{\em Commun. Math. Phys.} \textbf{323}, 635--661 (2013) 

\bibitem{LampartSchmidt2019}
J.~Lampart and J.~Schmidt. On Nelson-type Hamiltonians and abstract boundary conditions.
{\em Commun. Math. Phys.} \textbf{367}, 629--663 (2019) 


\bibitem{Lions1982}
P.~L.~Lions.
\newblock {\em Generalized solutions of Hamilton-Jacobi equations.}
\newblock Research Notes in Mathematics, vol.~69, Pitman, London (1982)

\bibitem{Matte2016}
O.~Matte. Continuity properties of the semi-group and its integral kernel in non-relativistic QED.
{\em Rev. Math. Phys.} \textbf{28}, 1650011, 90 pages (2016) 

\bibitem{MatteMoeller2018}
O.~Matte and J.~S.~M{\o}ller. Feynman--Kac formulas for the ultraviolet renormalized Nelson model.
{\em Ast\'{e}risque} \textbf{404}, vi+110 pages (2018) 

\bibitem{Nelson1964}
E.~Nelson.
{Interaction of nonrelativistic particles with a quantized scalar field}.
{\em J. Math. Phys.} \textbf{5}, 1190--1197 (1964) 

\bibitem{Panati2009}
A.~Panati. Existence and nonexistence of a ground state for the massless Nelson model under 
binding condition. {\em Rep. Math. Phys.} \textbf{63}, 305--330 (2009) 

\bibitem{Posilicano2020}
A.~Posilicano. On the self-adjointness of $H+A^*+A$.
{\em Math. Phys. Anal. Geom.} \textbf{23}, Art.~37, 31 pages (2020) 

\bibitem{Schmidt2021}
J.~Schmidt. The massless Nelson Hamiltonian and its domain. In: A.~Michelangeli (Editor).
{\em Mathematical challenges of zero-range physics.},  
Rome, July 2018, Springer INdAM Series, vol.~42, pages 57--80. Springer, Cham (2021)

\bibitem{spo98}
H.~Spohn.
\newblock {Ground state of a quantum particle coupled to a scalar boson field}.
\newblock {\em Lett. Math. Phys.} \textbf{44}, 9--16 (1998)

\bibitem{var67}
S.~R.~S. Varadhan.
\newblock Diffusion processes in a small time interval.
\newblock {\em Comm. Pure Appl. Math.} \textbf{20}, 659--685 (1967)

\end{thebibliography}
\end{document}